\documentclass[journal]{IEEEtran}
\IEEEoverridecommandlockouts
\usepackage{color}
\usepackage{graphicx}
\usepackage[dvips]{epsfig}
\usepackage{graphics} 
\usepackage{times} 
\usepackage[cmex10]{amsmath} 
\usepackage{amssymb}  
\usepackage{cite}
\usepackage{multirow}

\def\inftyn #1{\left\|#1\right\|_{\infty}}
\def\twon #1{\left\|#1\right\|_2}
\def\onen #1{\left\|#1\right\|_1}
\def\zeron #1{\left\|#1\right\|_0}
\def\frobn #1{\left\|#1\right\|_{\text{F}}}
\def\sgn #1{\text{sgn}#1}
\def\abs #1{\left|#1\right|}
\def\inp #1{\left\langle#1\right\rangle}

\def\rainfty{\rightarrow\infty}

\def\bR{\mathbb{R}}

\def\bib{\emph{\textbf{b}}}
\def\biA{\emph{\textbf{A}}}
\def\biB{\emph{\textbf{B}}}
\def\biPhi{{\Phi}}
\def\bix{\emph{\textbf{x}}}
\def\bip{\emph{\textbf{p}}}
\def\biz{\emph{\textbf{z}}}
\def\bie{\emph{\textbf{e}}}
\def\biw{\emph{\textbf{w}}}
\def\biv{\emph{\textbf{v}}}
\def\biu{\emph{\textbf{u}}}
\def\biy{\emph{\textbf{y}}}
\def\biI{\emph{\textbf{I}}}
\def\biH{\emph{\textbf{H}}}
\def\biAx{\emph{\textbf{A}}\emph{\textbf{x}}}

\def\cF{\mathcal{F}}

\def\cL{\mathcal{L}}

\def\cW{\mathcal{W}}

\def\bee{\begin{equation}}
\def\ene{\end{equation}}

\def\beq{\begin{eqnarray}}
\def\enq{\end{eqnarray}}
\def\lentwo{\setlength\arraycolsep{2pt}}

\newtheorem{lem}{Lemma}
\newtheorem{rem}{Remark}

\newtheorem{thm}{Theorem}

\def\equ #1{\begin{equation}#1\end{equation}}
\def\equa #1{\begin{eqnarray}#1\end{eqnarray}}
\def\sbra #1{\left(#1\right)}
\def\mbra #1{\left[#1\right]}
\def\lbra #1{\left\{#1\right\}}

\def\bp{\left(\text{BP}\right)}
\def\bpdn{\left(\text{BP}_\epsilon\right)}
\def\qp{\left(\text{QP}_\lambda\right)}
\def\bpo{\left(\text{BP}^o\right)}

\def\sp{\left(SP\right)}
\def\spx{\left(SP_1\right)}
\def\spe{\left(SP_2\right)}

\title{Orthonormal Expansion $\ell_1$-Minimization Algorithms for Compressed Sensing}
\author{Zai Yang, Cishen Zhang, Jun Deng, and Wenmiao Lu

\thanks{To appear in IEEE Transactions on Signal Processing. Copyright (c) 2011 IEEE. Personal use of this material is permitted. However, permission to use this material for any other purposes must be obtained from the IEEE by sending a request to pubs-permissions@ieee.org.}

\thanks{Z. Yang, J. Deng and W. Lu are with the School of Electrical and Electronic Engineering,
Nanyang Technological University, 639798, Singapore (e-mail: yang0248@e.ntu.edu.sg; de0001un@e.ntu.edu.sg; wenmiao@ntu.edu.sg).

C. Zhang is with the Faculty of Engineering and Industrial Sciences, Swinburne University of Technology, Hawthorn VIC 3122, Australia (e-mail: cishenzhang@swin.edu.au).}}
\begin{document}
\maketitle

\begin{abstract}
Compressed sensing aims at reconstructing sparse signals from
significantly reduced number of samples, and a popular reconstruction approach is $\ell_1$-norm minimization. In this correspondence, a method called orthonormal expansion is presented to reformulate the basis pursuit problem for noiseless compressed sensing. Two algorithms are proposed based on convex optimization: one exactly solves the problem and the other is a relaxed version of the first one. The latter can be considered as a modified iterative soft thresholding algorithm and is easy to implement. Numerical simulation shows that, in dealing with noise-free measurements of sparse signals, the relaxed version is accurate, fast and competitive to the recent state-of-the-art algorithms. Its practical application is demonstrated in a more general case where signals of interest are approximately sparse and measurements are contaminated with noise.
\end{abstract}
\begin{IEEEkeywords}
Augmented Lagrange multiplier, Compressed sensing, $\ell_1$
minimization, Orthonormal expansion, Phase transition,
Sparsity-undersampling tradeoff.
\end{IEEEkeywords}
\section{Introduction}
\IEEEPARstart{C}{ompressed} sensing (CS) aims at reconstructing a
signal from its significantly reduced measurements with a priori knowledge
that the signal is (approximately) sparse
\cite{candes2006robust,candes2007sparsity,donoho2006compressed}. In CS, the signal $\bix^o\in\bR^N$ of interest is acquired by taking measurements of the form
\[\bib=\biAx^o+\bie,\]
where $\biA\in\bR^{n\times N}$ is the sampling matrix, $\bib\in\bR^{n}$ is the vector of our measurements or samples, $\bie\in\bR^{n}$ is the measurement noise, with $N$ and $n$ being sizes of the signal and acquired samples, respectively. A standard approach to reconstructing the original signal $\bix^o$ is to solve
\[\bpdn\quad\min_\bix\onen{\bix}, \text{ subject to }\twon{\bib-\biAx}\leq\epsilon,\]
which is known as the basis pursuit denoising (BPDN) problem, with $\twon{\bie}\leq\epsilon$.
Another frequently discussed approach is to solve the problem in its Lagrangian form
\[\qp\quad\min_\bix\lbra{\lambda\onen{\bix}+\frac{1}{2}\twon{\bib-\biAx}^2}.\]
It follows from the knowledge of convex optimization that $\bpdn$ and $\qp$ are equivalent with appropriate choices of $\epsilon$ and $\lambda$. In general, $\lambda$ decreases as $\epsilon$ decreases. In the limiting case of $\lambda,\epsilon\rightarrow0$, both $\bpdn$ and $\qp$ converges to the following basis pursuit (BP) problem in noiseless CS:
\[\bp\quad\min_\bix\onen{\bix}, \text{ subject to }\biAx=\bib.\]

It is important to develop accurate and computationally efficient algorithms to deal with the $\ell_1$ minimization problems of high dimensional signals in CS, such as an image of $512\times512$ pixels. One popular approach for solving $\qp$ is the iterative soft thresholding (IST) algorithm of the form (stating from $\bix_0=0$) \cite{daubechies2004iterative,bredies2008linear}
\equa{\bix_{t+1}=S_{\lambda}\sbra{\bix_t+\biA'\biz_t},\quad \biz_{t}=\bib-\biAx_{t},\label{formula:standardIST}}
where $'$ denotes the transpose operator and $S_{\lambda}\sbra{\cdot}$ is the soft thresholding operator with a threshold $\lambda$, which will be defined more precisely in Section \ref{section_preliminaries}. IST has a concise form and is easy to implement, but its convergence can be very slow in some cases \cite{bredies2008linear}, especially for small $\lambda$. Other algorithms proposed to solve $\qp$ include interior-point method \cite{kim2008interior}, conjugate gradient method\cite{lustig2007sparse} and fixed-point continuation\cite{hale2007fixed}. Few algorithms can accurately solve large-scale BPDN problem $\bpdn$ with a low computational complexity. The $\ell_1$-magic package\cite{candes-l1} includes a primal log barrier code solving $\bpdn$, in which the conjugate gradient method may not find a precise Newton step in the large-scale mode. NESTA\cite{becker2009nesta} approximately solves $\bpdn$ based on Nesterov's smoothing technique\cite{nesterov2005smooth}, with continuation.

In the case of strictly sparse signals and noise-free measurements, many fast algorithms have been proposed to exactly reconstruct $\bix^o$. One class of algorithms uses the greedy pursuit method, which iteratively refines the support and entries of a sparse solution to yield a better approximation of $\bix^o$, such as OMP\cite{tropp2007signal}, StOMP\cite{donoho2006sparse} and CoSaMP\cite{needell2009cosamp}. These algorithms, however, may not produce satisfactory sparsity-undersampling tradeoff compared with $\bp$ because of their greedy operations. As mentioned before, $\qp$ is equivalent to $\bp$ as $\lambda\rightarrow0$. Hence, $\bp$ can be solved with high accuracy using algorithms for $\qp$ by setting $\lambda$ to a small value, e.g. $1\times10^{-6}$. IST has attracted much attention because of its simple form. In the case where $\lambda$ is small, however, the standard IST in (\ref{formula:standardIST}) can be very slow. To improve its speed, a fixed-point continuation (FPC) strategy is exploited \cite{hale2007fixed}, in which $\lambda$ is decreased in a continuation scheme and a $q$-linear convergence rate is achieved. Further, FPC-AS \cite{wen2010fast} is developed to improve the performance of FPC by introducing an active set, inspired by greedy pursuit algorithms. An alternative approach to improving the speed of IST is to use an aggressive continuation, which takes the form
\equa{\bix_{t+1}=S_{\lambda_t}\sbra{\bix_t+\biA'\biz_t},\quad \biz_{t}=\bib-\biAx_{t},\label{formula:greedyIST}}
where $\lambda_t$ may decrease in each iteration. The algorithm of this form typically has a worse sparsity-undersampling tradeoff than $\bp$\cite{maleki2010optimally}. Such a disadvantage is partially overcome by approximately message passing (AMP) algorithm\cite{donoho2009message}, in which a modification is introduced in the current residual $\biz_t$:
\equa{\bix_{t+1}=S_{\lambda_t}\sbra{\bix_t+\biA'\biz_t},\quad \biz_{t}=\bib-\biAx_{t}+\frac{N\zeron{\bix_t}}{n}\biz_{t-1},\label{formula:amp}}
where $\zeron{\bix}$ counts the number of nonzero entries of $\bix$. It is noted that AMP having the same spasity-undersampling tradeoff as $\bp$ is only established based on heuristic arguments and numerical simulations. Moreover, it cannot be easily extended to deal with more general complex-valued sparse signals, though real-valued signals are only considered in this correspondence.

This correspondence focuses on solving the basis pursuit problem $\bp$ in noiseless CS. We assume that $\biA\biA'$ is an identity matrix, i.e., the rows of $\biA$ are orthonormal. This is reasonable since most fast transforms in CS are of this form, such as discrete cosine transform (DCT), discrete Fourier transform (DFT) and some wavelet transforms, e.g. Haar wavelet transform. Such an assumption has also been used in other algorithms, e.g. NESTA. A novel method called orthonormal expansion is introduced to reformulate $\bp$ based on this assumption. The exact OrthoNormal Expansion $\ell_1$ minimization (eONE-L1) algorithm is then proposed to exactly solve $\bp$ based on the augmented Lagrange multiplier (ALM) method, which is a convex optimization method.

The relaxed ONE-L1 (rONE-L1) algorithm is further developed to simplify eONE-L1. It is shown that rONE-L1 converges at least exponentially and is in the form of modified IST in (\ref{formula:greedyIST}). In the case of strictly sparse signals and noise-free measurements, numerical simulations show that rONE-L1 has the same sparsity-undersampling tradeoff as $\bp$ does. Under the same conditions, rONE-L1 is compared with state-of-the-art algorithms, including FPC-AS, AMP and NESTA. It is shown that rONE-L1 is faster than AMP and NESTA when the number of measurements is just enough to exactly reconstruct the original sparse signal using $\ell_1$ minimization. In a general case of approximately sparse signals and noise-contaminated measurements, where AMP is omitted for its poor performance, an example of 2D image reconstruction from its partial-DCT measurements demonstrates that rONE-L1 outperforms NESTA and FPC-AS in terms of computational efficiency and reconstruction quality, respectively.

The rest of the correspondence is organized as follows. Section
\ref{section_algorithms} introduces the proposed exact and relaxed ONE-L1
algorithms followed by their implementation details.
Section \ref{section:simulation} reports the efficiency of our algorithm through numerical simulations in both noise-free and noise-contaminated cases. Conclusions are drawn in Section \ref{section_conclusion}.

\section{ONE-L1 Algorithms}\label{section_algorithms}
\subsection{Preliminary: Soft Thresholding Operator}\label{section_preliminaries}

For $w\in \bR$, the soft thresholding of $w$ with a threshold
$\lambda \in \bR^+$ is defined as:
\begin{equation}S_\lambda (w)=\sgn\sbra{w}\cdot(\abs{w}-\lambda)^+,\notag
\end{equation}
where $(\cdot)^+=\max(\cdot,0)$ and
\begin{equation}
\sgn\sbra{w}=\left\{
\begin{array}{ll}w/\abs{w},&w\neq0;\\0,&w=0.
\end{array}\right.\notag
\end{equation}
The operator $S_\lambda(\cdot)$ can be extended to vector variables
by its element-wise operation.

The soft thresholding operator can be applied to solve the
following $\ell_1$-norm regularized least square problem \cite{daubechies2004iterative}, i.e.,
\equ{S_\lambda (\biw)=\arg\min_\biv \lbra{\lambda
\onen{\biv}+\frac{1}{2}\twon{\biw-\biv}^2}.\label{soft_threshold}}
where $\biv,\biw\in \bR^n$, and $S_\lambda (\biw)$ is the unique minimizer.

\subsection{Problem Description}
Consider the $\ell_1$-minimization problem $\bp$ with the sampling matrix $\biA$ satisfying that $\biA\biA'=\biI$, where $\biI$ is an identity matrix. We call that $\biA$ is a partially orthonormal matrix hereafter as its rows are usually randomly selected from an orthonormal matrix in practice, e.g. partial-DCT matrix. Hence, there exists
another partially orthonormal matrix $\biB\in\bR^{(N-n)\times N}$,
whose rows are orthogonal to those of $\biA$, such that
$\biPhi=\begin{bmatrix}\biA\\\biB\end{bmatrix}$
is orthonormal. Let $\bip=\biPhi \bix$. The problem $\bp$ is then equivalent to
\[\bpo\quad\min_{\bix,\bip,\Gamma(\bip)=\bib}\onen{\bix}, \text{ subject to }\biPhi \bix=\bip,\]
where $\Gamma(\bip)$ is an operator projecting the vector $\bip$ onto its first $n$ entries.

In $\bpo$, the sampling matrix $\biA$ is expanded into an orthonormal matrix $\biPhi$. It corresponds to the scenario where the full sampling is carried out with the sampling matrix $\biPhi$ and $\bip$ is the vector containing all measurements. Note that only $\bib$, as a part of $\bip$, is actually observed. To expand the sampling matrix $\biA$ into an orthonormal matrix $\biPhi$ is a key step to show that the ALM method exactly solves $\bpo$ and, hence, $\bp$. The next subsection describes the proposed algorithm, referred to as orthonormal expansion $\ell_1$-minimization.

\subsection{ALM Based ONE-L1 Algorithms}

In this subsection we solve the $\ell_1$-minimization problem $\bpo$ using the ALM method\cite{nocedal1999numerical,bertsekas1982constrained}. The ALM method is similar to the quadratic penalty method except an additional Lagrange multiplier term. Compared with the quadratic penalty method, ALM method has some salient properties, e.g. the ease of parameter tuning and the convergence speed. The augmented Lagrangian function is
\equ{\cL(\bix,\bip,\biy,\mu)=\onen{\bix}+\inp{\bip-\biPhi \bix,\biy}+\frac{\mu}{2}\twon{\bip-\biPhi \bix}^2, \label{formula_lagrange_function}}
where Lagrange multiplier $\biy\in\bR^N$, $\mu\in\bR^+$ and $\inp{\biu,\biv}=\biu'\biv\in\bR$
is the inner product of $\biu,\textrm{ }\biv\in\bR^N$. Eq. (\ref{formula_lagrange_function}) can be expressed as follows:
\equ{\cL(\bix,\bip,\biy,\mu)=\onen{\bix}+\frac{\mu}{2}\twon{\bip-\biPhi \bix+\mu^{-1}\biy}^2-\frac{1}{2\mu}\twon{\biy}^2. \label{formula_lagrange_variation}}
Subsequently, we have the following optimization problem $\sp$:
\[\sp\quad\min_{\bix,\bip,\Gamma(\bip)=\bib}\cL(\bix,\bip,\biy,\mu).\]
Instead of solving $\sp$, let us consider the two related problems
\[\spx\quad\min_\bix \cL(\bix,\bip,\biy,\mu),\] and
\[\spe\quad\min_{\bip,\Gamma(\bip)=\bib} \cL(\bix,\bip,\biy,\mu).\]
Note that problem $\spx$ is similar to the $\ell_1$-regularized problem in (\ref{soft_threshold}). In general, $\spx$ cannot be directly solved using the soft thresholding operator as that in (\ref{soft_threshold}) since there is a matrix product of $\Phi$ and $\bix$ in the term of $\ell_2$-norm. However, the soft thresholding operator does apply to the special case where $\biPhi$ is orthonormal. Given $\twon{\biPhi \biu}=\twon{\biu}$ for any $\biu\in\bR^N$, we can apply the soft thresholding to obtain
\equ{S_{\mu^{-1}}\sbra{\biPhi'\sbra{\bip+\mu^{-1}\biy}}=\arg\min_\bix \cL(\bix,\bip,\biy,\mu). \label{formula_update_x}}
 To solve $\spe$, we let $\partial_{\overline{\Gamma}(\bip)} \cL(\bix,\bip,\biy,\mu)=0$ to obtain $\overline{\Gamma}(\bip)=\overline{\Gamma}\sbra{\biPhi \bix-\mu^{-1}\biy}$,
i.e.,
\equ{\mbra{\lentwo{\begin{array}{c}\bib\\\overline{\Gamma}\sbra{\biPhi \bix-\mu^{-1}\biy}\end{array}}}=\arg\min_{\bip,\Gamma(\bip)=b} \cL(\bix,\bip,\biy,\mu), \label{formula_update_p}}
where $\overline\Gamma\sbra{\cdot}$ is the operator projecting the variable to its last $N-n$ entries. As a result, an iterative solution of $\sp$ is stated in the following Lemma \ref{lemma_alternating_sp}.

\begin{lem} For fixed $\biy$ and $\mu$, the iterative algorithm given by
\equa{
&& \bix^{j+1}=S_{\mu^{-1}}\sbra{\biPhi'\sbra{\bip^j+\mu^{-1}\biy}}, \label{formula_sp_x_iter}\\
&& \bip^{j+1}=\mbra{\lentwo{\begin{array}{c}\bib\\\overline{\Gamma}\sbra{\biPhi \bix^{j+1}-\mu^{-1}\biy}\end{array}}} \label{formula_sp_p_iter}}
converges to an optimal solution of $\sp$. \label{lemma_alternating_sp}
\end{lem}

\begin{proof}
 Denote $\cL\sbra{\bix,\bip,\biy,\mu}$ as $\cL\sbra{\bix,\bip}$, for simplicity. By the optimality and uniqueness of $\bix^{j+1}$ and $\bip^{j+1}$, we have $\cL\sbra{\bix^{j+1},\bip^{j+1}}\leq\cL\sbra{\bix^{j},\bip^{j}}$ and the equality holds if and only if $\sbra{\bix^{j+1},\bip^{j+1}}=\sbra{\bix^j,\bip^j}$. Hence, the sequence $\lbra{\cL\sbra{\bix^{j},\bip^{j}}}$ is bounded and converges to
 a constant $L^*$, i.e., $\cL\sbra{\bix^{j},\bip^{j}}\rightarrow L^*$ as $j\rightarrow +\infty$. Since the sequence $\lbra{\bix^j}$ is also bounded by $\onen{\bix^j}\leq\cL\sbra{\bix^j,\bip^j}+\frac{1}{2\mu}\twon{\biy}^2$, there exists a sub-sequence $\lbra{\bix^{j_i}}_{i=1}^{+\infty}$ such that $\bix^{j_i}\rightarrow \bix_s^*$ as $i\rightarrow+\infty$, where $\bix_s^*$ is an accumulation point of $\lbra{\bix^j}$. Correspondingly, $\bip^{j_i}\rightarrow \bip_s^*$, and moreover, $\cL\sbra{\bix_s^*,\bip_s^*}=L^*$.


We now use contradiction to show that $\sbra{\bix_s^*,\bip_s^*}$ is a fixed point of the algorithm.
Suppose that there exist $\sbra{\overline \bix_s,\overline \bip_s}\neq\sbra{\bix_s^*,\bip_s^*}$ such that $\overline \bix_s=\arg\min_\bix\cL\sbra{\bix,\bip_s^*}$ and $\overline \bip_s=\arg\min_{\bip,\Gamma\sbra{\bip}=\bib}\cL\sbra{\bix_s^*,\bip}$, and $\cL\sbra{\overline \bix_s,\overline \bip_s}<L^*$. By (\ref{formula_sp_x_iter}), (\ref{formula_sp_p_iter}) and \cite[Lemma 2.2]{daubechies2004iterative}, we have $\twon{\bix^{j_i+1}-\overline \bix_s}\leq\twon{\bix^{j_i}-\bix_s^*}\rightarrow0$, i.e., $\bix^{j_i+1}\rightarrow\overline \bix_s$, as $i\rightarrow+\infty$. Meanwhile, $\bip^{j_i+1}\rightarrow\overline \bip_s$. Hence, $\cL\sbra{\bix^{j_i+1},\bip^{j_i+1}}\rightarrow\cL\sbra{\overline \bix_s,\overline \bip_s}<L^*$, which contradicts $\cL\sbra{\bix^{j},\bip^{j}}\rightarrow L^*$, resulting in that $\sbra{\bix_s^*,\bip_s^*}$ is a fixed point. Moreover, it follows from $\twon{\bix^{j_i+q}-\bix_s^*}\leq\twon{\bix^{j_i}-\bix_s^*}\rightarrow0$ for any positive integer $q$, that $\bix^j\rightarrow \bix_s^*$, as $j\rightarrow+\infty$.

Note that orthonormal matrix $\biPhi=\begin{bmatrix}\biA \\\biB\end{bmatrix}$ and $\biPhi'\biPhi=\biA'\biA+\biB'\biB=\biI$. We can obtain
\equ{\bix_s^*=S_{\mu^{-1}}\sbra{\bix_s^*+\biA'\sbra{\bib+\mu^{-1}\Gamma(\biy)-\biAx_s^*}}.}
Meanwhile, $\sp$ is equivalent to
\equ{\begin{split}\min_\bix
& \cL\sbra{\bix,\mbra{\lentwo\begin{array}{c}\bib\\\overline\Gamma\sbra{\biPhi \bix-\mu^{-1}\biy}\end{array}}}\\
& =\onen{\bix}+\frac{\mu}{2}\twon{\biAx-\bib-\mu^{-1}\Gamma(\biy)}^2-\frac{1}{2\mu}\twon{\biy}^2.\end{split}\label{formula_sp_simplified}}
By \cite[Proposition 3.10]{daubechies2004iterative}, $\bix_s^*$ is an optimal solution of the problem in (\ref{formula_sp_simplified}) and equivalently, $\sbra{\bix_s^*,\bip_s^*}$ is an optimal solution of $\sp$.
\end{proof}

\begin{rem}~ \label{remark_sp}
\begin{itemize}
 \item[(1)] Lemma 1 shows that to solve problem $\sp$ is equivalent to solve, iteratively, problems $\spx$ and $\spe$.
 \item[(2)]Reference \cite[Proposition 3.10]{daubechies2004iterative} only deals with the special case $\twon{\biA}<1$ and it is, in fact, straightforward to extend the result to arbitrary $\biA$.
     \end{itemize}
\end{rem}

Following from the framework of the ALM method \cite{bertsekas1982constrained} and Lemma \ref{lemma_alternating_sp}, the ALM based ONE-L1 algorithm is outlined in Algorithm 1, where $\sbra{\bix_t^*,\bip_t^*}$ is the optimal solution to $\sp$ in the $t$th iteration and $\biy_t^*$ is the corresponding Lagrange multiplier.

\smallskip

{\flushleft
\footnotesize
\begin{tabular}{@{} r @{ } l @{}}
\hline
\multicolumn{2}{l}{Algorithm 1: Exact ONE-L1 Algorithm via ALM Method}\\
\hline \hline
\multicolumn{2}{l}{Input: Expanded orthonormal matrix $\biPhi$ and observed sample data $\bib$.}\\
1. & $\bix_0^*=\textbf{\emph{0}}$; $\setlength\arraycolsep{2pt}\bip_0^*=\left[\begin{array}{c}\bib\\\emph{\textbf{0}}\end{array}\right]$; $\biy_0^*=\textbf{\emph{0}}$; $\mu_0>0$; $t=0$.\\
2. & while not converged do\\
3. & \quad Lines 4-9 solve $\left(\bix_{t+1}^*,\bip_{t+1}^*\right)=\arg\min_{\left(\bix,\bip,\Gamma(\bip)=\bib\right)}\cL\left(\bix,\bip,\biy_t^*,\mu_t\right)$;\\
4. & \quad $\bix_{t+1}^0=\bix_t^*$, $\bip_{t+1}^0=\bip_t^*$, $j=0$;\\
5. & \quad while not converged do\\
6. & \quad\quad $\bix_{t+1}^{j+1}=S_{\mu_t^{-1}}\left(\biPhi'\left(\bip_{t+1}^j+\mu_t^{-1}\biy_t^*\right)\right)$;\\
7. & \quad\quad $\bip_{t+1}^{j+1}=\mbra{\lentwo{\begin{array}{c}\bib\\\overline{\Gamma}\sbra{\biPhi \bix_{t+1}^{j+1}-\mu_t^{-1}\biy_t^*}\end{array}}}$;\\
8. & \quad\quad set $j=j+1$;\\
9. & \quad end while\\
10.& $\biy_{t+1}^*=\biy_t^*+\mu_t\sbra{\bip_{t+1}^*-\biPhi \bix_{t+1}^*}$;\\
11.& choose $\mu_{t+1}>\mu_t$;\\
12.& set $t=t+1;$\\
13.& end while\\
\multicolumn{2}{l}{Output: $\left(\bix_t^*,\bip_t^*\right)$.}\\
\hline
\end{tabular}
}
\smallskip

The convergence of Algorithm 1 is stated in the following theorem.
\begin{thm}\label{theorem_exact_ONE_L1}
Any accumulation point $\left(\bix^*,\bip^*\right)$ of sequence $\lbra{\left(\bix_t^*,\bip_t^*\right)}_{t=1}^{+\infty}$ of Algorithm 1
is an optimal solution of $\bpo$ and the convergence rate with respect to the outer iteration loop index $t$ is at least $O\sbra{\mu_{t-1}^{-1}}$ in the sense that
\[\abs{\onen{\bix_t^*}-\bix^{\dagger}}=O\left(\mu_{t-1}^{-1}\right),\]
where $\bix^{\dagger}=\|\bix^*\|_1$.
\end{thm}

\begin{proof} \label{section_proof_theorem_exact} We first show that the sequence $\lbra{\biy_{t}^*}$ is bounded. By the optimality of $\left(\bix_{t+1}^*,\bip_{t+1}^*\right)$ we have
{\setlength\arraycolsep{2pt}
\equa{0 &\in & \partial _\bix\cL\left(\bix_{t+1}^*,\bip_{t+1}^*,\biy_t^*,\mu_t\right)=\partial\onen{\bix_{t+1}^*}-\biPhi'\biy_{t+1}^*,\notag\\
      0 & = & \partial_{\overline\Gamma\sbra{\bip}}\cL\left(\bix_{t+1}^*,\bip_{t+1}^*,\biy_t^*,\mu_t\right)
      =\overline{\Gamma}\left(\biy_{t+1}^*\right),\notag}
}where $\partial_{\bix}$ denotes the partial differential operator with respect to $\bix$ resulting in a set of subgradients. Hence, $\biPhi'\biy_{t+1}^*\in\partial\onen{\bix_{t+1}^*}$. It follows that $\inftyn{\biPhi'\biy_{t+1}^*}\leq1$ and $\lbra{\biy_{t}^*}$ is bounded. By $\bix^{\dagger}\geq\cL\sbra{\bix_{t+1}^*,\bip_{t+1}^*,\biy_t^*,\mu_t}$,
\equ{\begin{split}
\onen{\bix_{t+1}^*} &=\cL\sbra{\bix_{t+1}^*,\bip_{t+1}^*,\biy_t^*,\mu_t}-\frac{1}{2\mu_t}\sbra{\twon{\biy_{t+1}^*}^2-\twon{\biy_{t}^*}^2}\notag\\
&\leq \bix^{\dagger}-\frac{1}{2\mu_t}\sbra{\twon{\biy_{t+1}^*}^2-\twon{\biy_{t}^*}^2\notag}.
\end{split}}
By $\lbra{\biy_{t}^*}$ is bounded,
\equ{\onen{\bix_{t+1}^*}\leq \bix^{\dagger}+O\sbra{\mu_t^{-1}}.\label{formula_upperbound_onenorm}}
For any accumulation point $\bix^*$ of $\bix_t^*$, without loss of generality, we have $\bix_t^*\rightarrow \bix^*$ as $t\rightarrow+\infty$. Hence, $\onen{\bix^*}\leq \bix^{\dagger}$.
In the mean time, $\bip_{t+1}^*=\biPhi \bix_{t+1}^*+\mu_t^{-1}\sbra{\biy_{t+1}^*-\biy_t^*}\rightarrow \bip^*$ and $\biPhi \bix^*=\bip^*$ result in that $\sbra{\bix^*,\bip^*}$ is an optimal solution to $\bpo$.

Moreover, by $\bix_{t+1}^*=\biPhi'\mbra{\bip_{t+1}^*-\mu_t^{-1}\sbra{\biy_{t+1}^*-\biy_t^*}}$ and
\equ{\bix^{\dagger} =\min_{\biPhi \bix=\bip,\Gamma(\bip)=\bib}\onen{\bix}=\min_{\bip,\Gamma(\bip)=\bib}\onen{\biPhi'\bip}\leq\onen{\biPhi'\bip_{t+1}^*},\notag}
we have $\onen{\bix_{t+1}^*}\geq \bix^{\dagger}-O\sbra{\mu_t^{-1}}$, which establishes the theorem with (\ref{formula_upperbound_onenorm}).
\end{proof}

Algorithm 1 contains, respectively, an inner and an outer iteration loops. Theorem \ref{theorem_exact_ONE_L1} presents only the convergence rate of the outer loop. A natural way to speed up Algorithm 1 is to terminate the inner loop without convergence and use the obtained inner-loop solution as the initialization for the next iteration. This is similar to a continuation strategy and can be realized with reasonably set precision and step size $\mu_t$ \cite{bertsekas1982constrained}. When the continuation parameter $\mu_t$ increases very slowly, in a few iterations, the inner loop can produce a solution with high accuracy. In particular, for the purpose of fast and simple computing, we may update the variables in the inner loop only once before stepping into the outer loop operation. This results in a relaxed version of exact ONE-L1 algorithm (eONE-L1), namely relaxed ONE-L1 algorithm (rONE-L1) outlined in Algorithm 2.

{\flushleft
\footnotesize
\begin{tabular}{@{} c @{ } l @{}}
\hline
\multicolumn{2}{l}{Algorithm 2: Relaxed ONE-L1 Algorithm}\\
\hline \hline
\multicolumn{2}{l}{Input: Expanded orthonormal matrix $\biPhi$ and observed sample data $\bib$.}\\
1. & $\bix_0=\emph{\textbf{0}}$; $\setlength\arraycolsep{2pt}\bip_0=\left[\begin{array}{c}\bib\\\emph{\textbf{0}}\end{array}\right]$; $\biy_0=\emph{\textbf{0}}$; $\mu_0>0$; $t=0$.\\
2. & while not converged do\\
3. & \quad $\bix_{t+1}=S_{\mu_t^{-1}}\left(\biPhi'\left(\bip_{t}+\mu_t^{-1}\biy_t\right)\right)$;\\
4. & \quad $\bip_{t+1}=\mbra{\lentwo{\begin{array}{c}\bib\\\overline{\Gamma}\sbra{\biPhi \bix_{t+1}-\mu_t^{-1}\biy_t}\end{array}}}$;\\
5. & \quad $\biy_{t+1}=\biy_t+\mu_t\sbra{\bip_{t+1}-\biPhi \bix_{t+1}}$;\\
6. & choose $\mu_{t+1}>\mu_t$;\\
7. & set $t=t+1;$\\
8. & end while\\
\multicolumn{2}{l}{Output: $\left(\bix_t,\bip_t\right)$.}\\
\hline
\end{tabular}
}
\smallskip

\begin{thm}\label{theorem_relaxed_ONE_L1}
The iterative solution $(\bix_t,\bip_t)$ of Algorithm 2 converges to a feasible solution $(\bix^f,\bip^f)$ of $\bpo$ if $\sum_{t=1}^{+\infty}\mu_t^{-1}<+\infty$. It converges at least exponentially to $(\bix^f,\bip^f)$ if $\lbra{\mu_t}$ is an exponentially increasing sequence.
\end{thm}

\begin{proof} \label{section_proof_theorem_relaxed}
We show first that sequences $\left\{\hat \biy_t\right\}$ and $\left\{\biy_t\right\}$ are bounded, where $\hat \biy_t=\biy_{t-1}+\mu_{t-1}\sbra{\bip_{t-1}-\biPhi \bix_t}$. By the optimality of $\bix_{t+1}$ and $\bip_{t+1}$ we have
{\lentwo\equa{
0 &\in& \partial_\bix\cL\left(\bix_{t+1},\bip_t,\biy_t,\mu_t\right)=\partial\onen{\bix_{t+1}}-\biPhi'\hat \biy_{t+1},\notag\\
0 & = & \partial_{\overline\Gamma\sbra{\bip}}\cL\left(\bix_{t+1},\bip_{t+1},\biy_t,\mu_t\right)\notag
=\overline{\Gamma}\left(\biy_{t+1}\right).\notag}
}Hence, $\inftyn{\biPhi'\hat \biy_{t+1}}\leq1$ and it follows that $\left\{\hat \biy_t\right\}$ is bounded.
Since $\biy_{t+1}=\hat \biy_{t+1}+\mu_t\sbra{\bip_{t+1}-\bip_t}$, we obtain
$\Gamma\sbra{\biy_{t+1}}=\Gamma\sbra{\hat \biy_{t+1}}$. This together with $\overline{\Gamma}\left(\biy_{t+1}\right)=0$
results in $\twon{\biy_{t+1}}\leq\twon{\hat \biy_{t+1}}$ and the boundedness of $\lbra{\biy_t}$. By $\bip_{t+1}-\bip_t=\mu_t^{-1}\sbra{\biy_{t+1}-\hat \biy_{t+1}}$, we have $\twon{\bip_{t+1}-\bip_t}\leq C\mu_t^{-1}$ with $C$ being a constant. Then $\lbra{\bip_t}$ is a Cauchy sequence if $\sum_{t=1}^{+\infty}\mu_t^{-1}<+\infty$, resulting in $\bip_t\rightarrow \bip^f$ as $t\rightarrow+\infty$. In the mean time, $\bix_t\rightarrow \bix^f$, $\biPhi \bix^f=\bip^f$. Hence, $\sbra{\bix^f,\bip^f}$ is a feasible solution of $\bpo$. Suppose that $\lbra{\mu_t}$ is an exponentially increasing sequence, i.e., $\mu_{t+1}=r\mu_t$ with $r>1$. By the boundedness of $\lbra{\biy_t}$ and $\lbra{\hat \biy_t}$ we have
\equ{\begin{split}\twon{\bip_t-\bip^f}
&=\twon{\sum_{i=t}^{+\infty}\sbra{\bip_i-\bip_{i+1}}}\leq\sum_{i=t}^{+\infty}\twon{\bip_i-\bip_{i+1}}\\
&\leq C\mu_t^{-1}\sum_{i=0}^{+\infty}r^{-i}=O\sbra{\mu_t^{-1}}.\end{split}\notag}
Hence, $\lbra{\bip_t}$ converges at least exponentially to $\bip^f$ since $\lbra{\mu_t^{-1}}$ exponentially converges to $0$, and the same result holds for $\lbra{\bix_t}$.
\end{proof}

\begin{rem}
It is shown in Theorem 2 that faster growth of $\lbra{\mu_t}$ can result in faster convergence of $\lbra{\bix_t}$. Intuitively, the reduced number of iterations for the inner loop problem $\sp$ may result in some error from the optimal solution $x^*_t$ of the inner loop. This will likely affect the accuracy of the final solution $x^f$ for $\bp$. Therefore, the growth speed of $\lbra{\mu_t}$ provides a tradeoff between the convergence speed of the algorithm and the precision of the final solution, which will be illustrated in Section \ref{section:simulation} through numerical simulations.
\end{rem}

\subsection{Relationship Between rONE-L1 and IST}

The studies and applications of IST type algorithms have been very active in recent years because of their concise presentations. This subsection considers the relationship between rONE-L1 and IST. Note that $\overline\Gamma\sbra{\biy_t}=0$ in Algorithm 2 and $\biPhi'\biPhi=\biA'\biA+\biB'\biB=\biI$. After some derivations, it can be shown that the rONE-L1 algorithm is equivalent to the following iteration (starting from $\bix_t=0$, as $t\leq0$, and $\biz_t=0$, as $t<0$):
\equa{\begin{split}\bix_{t+1}
&=S_{\lambda_t}\sbra{\bix_t+\biA'\biz_t},\\\biz_{t}
&=\bib-\biA\mbra{\sbra{1+\kappa_t}\bix_t-\kappa_t \bix_{t-1}}+\kappa_t \biz_{t-1},\end{split}\label{formula:rONE_L1}}
where $\lambda_t=\mu_t^{-1}$ and $\kappa_t=\frac{\mu_{t-1}}{\mu_{t}}$. Compared with the general form of IST in (\ref{formula:greedyIST}), one more term $\kappa_t \biz_{t-1}$ is added when computing the current residual $\biz_t$ in rONE-L1. Moreover, a weighted sum $(1+\kappa_t)\bix_t-\kappa_t \bix_{t-1}$ is used instead of the current solution $\bix_t$. It will be shown later that these two changes essentially improve the sparsity-undersampling tradeoff.

\begin{rem}Equations in (\ref{formula:rONE_L1}) show that the expansion from the partially orthonormal matrix $\biA$ to orthonormal $\biPhi$ is not at all involved in the actual implementation and computation of rONE-L1. The same claim also holds for eONE-L1 algorithm. Nevertheless, the orthonormal expansion is a key instrumentation in the derivation and analysis of Algorithms 1 and 2.\label{rem_expansion}
\end{rem}


\subsection{Implementation Details}\label{section_implementation}
As noted in Remark \ref{rem_expansion}, the expansion from $\biA$ to $\biPhi$ is not involved in the computing of ONE-L1 algorithms.
In our implementations, we consider using exponentially increasing $\mu_t$, i.e., fixing $r>1$ and $\mu_{t}=r^t\mu_0$. Let $Q_\alpha(\cdot)$ be an $\alpha$-quantile operator and $\mu_0=1/Q_\alpha\left(\abs{\biA'\bib}\right)$, with $\abs{\cdot}$ applying to the vector variable elementwise, $\mu_0^{-1}$ being the threshold in the first iteration and $\alpha=0.99$. In eONE-L1, a large $r$ can speed up the convergence of the outer loop iteration according to Theorem \ref{theorem_exact_ONE_L1}. However, simulations show that a larger $r$ can result in more iterations in the inner loop. We use $r=1+n/N$ as default. In rONE-L1, the value of $r$ provides a tradeoff between the convergence speed of the algorithm and the precision of the final solution. Our recommendation of $r$ to achieve the optimal sparsity-undersampling tradeoff is $r=\min\sbra{1+0.04n/N,1.02}$, which will be illustrated in Section \ref{section:simulation:tradeoff}.

An iterative algorithm needs a termination criterion. The eONE-L1 algorithm is considered converged if $\frac{\twon{\biAx_t^*-\bib}}{\twon{b}}<\tau_1$ with $\tau_1$ being a user-defined tolerance. The inner iteration is considered converged if $\frac{\twon{\bix_{t}^{j+1}-\bix_{t}^j}}{\twon{\bix_{t}^j}}<\tau_2$. In our implementation, the default values are $\sbra{\tau_1,\tau_2}=\sbra{10^{-5},10^{-6}}$. The rONE-L1 algorithm is considered converged if $\frac{\twon{\biAx_t-\bib}}{\twon{\bib}}<\tau$, with $\tau=10^{-5}$ as default.

\section{Numerical Simulations}\label{section:simulation}

\subsection{Sparsity-Undersampling Tradeoff}\label{section:simulation:tradeoff}
This subsection considers the sparsity-undersampling tradeoff of rONE-L1 in the case of strictly sparse signals and noise-free measurements. Phase transition is a measure of the sparsity-undersampling tradeoff in this case. Let the sampling ratio be $\delta=n/N$ and the sparsity ratio be $\rho=k/n$, where $k$ is a measure of sparsity of $\bix$, and we call that $\bix$ is $k$-sparse if at most $k$ entries of $\bix$ are nonzero. As $k,n,N\rainfty$ with fixed $\delta$ and $\rho$, the behavior of the phase transition of $\bp$ is controlled by $(\delta,\rho)$\cite{donoho2010counting,stojnic2009various}. We denote this theoretical curve by $\rho=\rho_T\sbra{\delta}$, which is plotted in Fig \ref{fig_pt_comparison_IST}.

We estimate the phase transition of rONE-L1 using a Monte Carlo method as in \cite{donoho2009observed,donoho2009message}. Two matrix ensembles are considered, including Gaussian with $N=1000$ and partial-DCT with $N=1024$. Here the finite-$N$ phase transition is defined as the value of $\rho$ at which the success probability to recover the original signal is $50\%$. We consider $33$ equispaced values of $\delta$ in $\lbra{0.02,0.05,\cdots,0.98}$. For each $\delta$, $21$ equispaced values of $\rho$ are generated in the interval $\mbra{\rho_T\sbra{\delta}-0.1, \rho_T\sbra{\delta}+0.1}$. Then $M=20$ random problem instances are generated and solved with respect to each combination of $(\delta,\rho)$, where $n=\lceil\delta N\rceil$, $k=\lceil\rho n\rceil$, and nonzero entries of sparse signals are generated from the standard Gaussian distribution. Success is declared if the relative root mean squared error (relative RMSE) $\frac{\twon{\hat \bix-\bix^o}}{\twon{\bix^o}}<10^{-4}$, where $\hat \bix$ is the recovered signal. The number of success among $M$ experiments is recorded. Finally, a generalized linear model is used to estimate the phase transition as in \cite{donoho2009observed}.

The experiment result is presented in Fig. \ref{fig_pt_comparison_IST}. The observed phase transitions using the recommended value of $r$ strongly agree with the theoretical result of $\bp$. It shows that the rONE-L1 algorithm has the optimal sparsity-undersampling tradeoff in the sense of $\ell_1$ minimization.

\begin{figure}
  \centering
  \includegraphics[height=2.6in, width=3.5in]{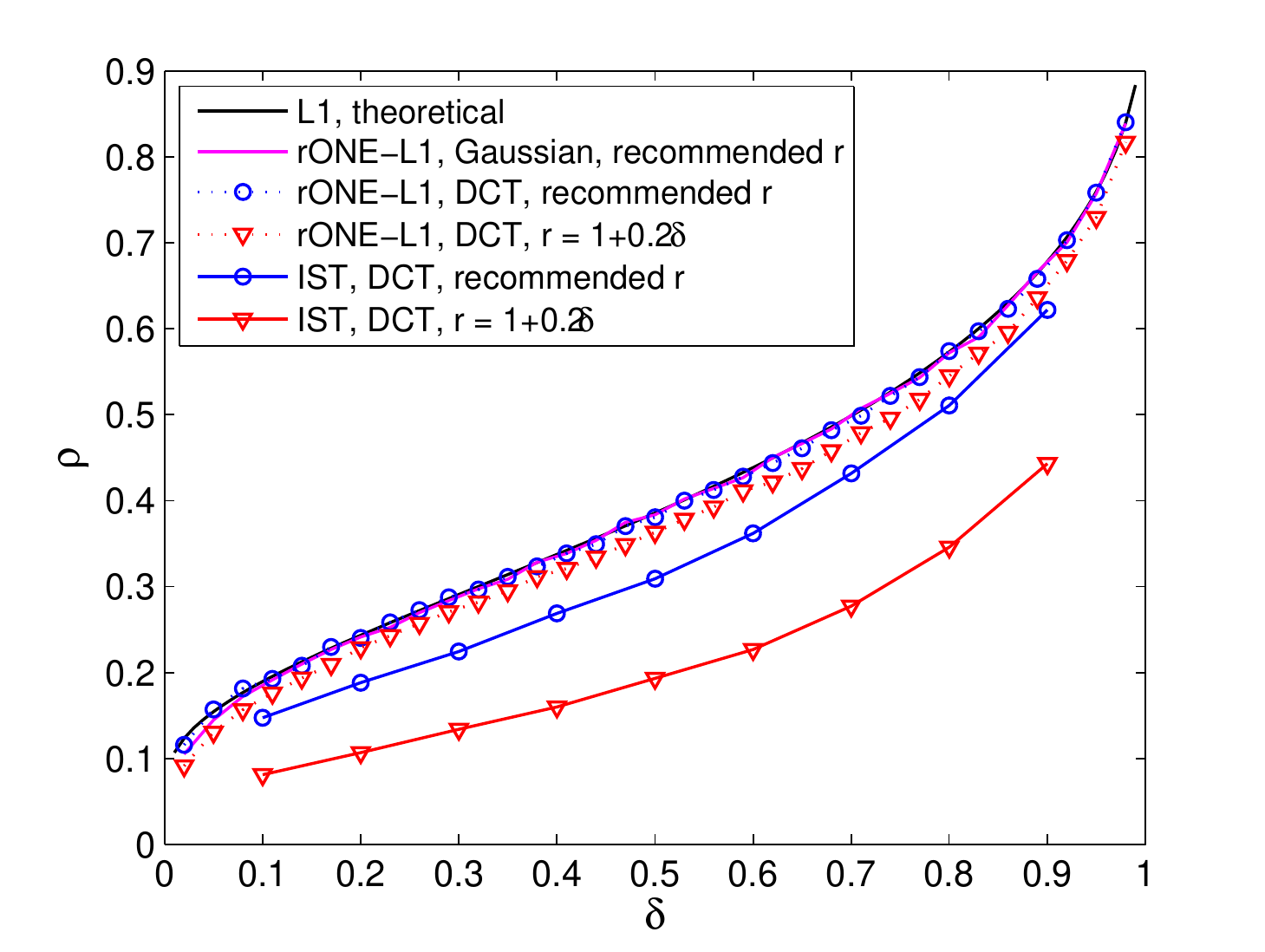}
  \caption{Observed phase transitions of rONE-L1, and comparison with those of IST. Note that, 1) the observed phase transitions of rONE-L1 with the recommended $r$ strongly agree with the theoretical calculation based on $\bp$; 2) rONE-L1 has significantly enlarged success phases compared with IST.}
  \label{fig_pt_comparison_IST}
\end{figure}

\subsection{Comparison with IST}\label{section:simulation:compIST}

The rONE-L1 algorithm can be considered as a modified version of IST in (\ref{formula:greedyIST}). In this subsection we compare the sparsity-undersampling tradeoff and speed of these two algorithms. A similar method is adopted to estimate the phase transition of IST, which is implemented using the same parameter values as rONE-L1. Only nine values of $\delta$ in $\lbra{0.1,0.2,\cdots,0.9}$ are considered with the partial-DCT matrix ensemble for time consideration. Another choice of $r=1+0.2\delta$ is considered besides the recommended one. Correspondingly, the phase transition of rONE-L1 with $r=1+0.2\delta$ is also estimated.

The observed phase transitions are shown in Fig. \ref{fig_pt_comparison_IST}. As a modified version of IST, obviously, rONE-L1 makes a great improvement over IST in the sparsity-undersampling tradeoff. Meanwhile, comparison of the averaged number of iterations of the two algorithms shows that rONE-L1 is also faster than IST. For example, as $\delta=0.2$ and the recommended $r$ is used, rONE-L1 is about 6 times faster than IST.

\subsection{Comparison with AMP, FPC-AS and NESTA in Noise-free Case}

In this subsection, we report numerical simulation results comparing rONE-L1 with state-of-the-art algorithms, including AMP, FPC-AS and NESTA, in the case of sparse signals and noise-free measurements. Our experiments used FPC-AS v.1.21, NESTA v.1.1, and AMP codes provided by the author. We choose parameter values for FPC-AS and NESTA such that each method produces a solution with approximately the same precision as that produced by rONE-L1. In our experiments we consider the recovery of exactly sparse signals from partial-DCT measurements. We set $N=2^{14}$ and $\delta=0.2$, and an `easy' case where $\rho=0.1$ and a `hard' case where $\rho=0.22$ are considered, respectively. \footnote{Here `easy' and `hard' refer to the difficulty degree in recovering a sparse signal from a specific number of measurements. The setting $\sbra{\delta,\rho}=(0.2,0.22)$ is close to the phase transition of $\bp$.} Twenty random problems are created and solved for each algorithm with each combination of $\sbra{\delta,\rho}$, and the minimum, maximum and averaged relative RMSE, number of calls of $\biA$ and $\biA'$, and CPU time usages are recorded. All experiments are carried on Matlab v.7.7.0 on a PC with a Windows XP system and a 3GHz CPU. Default parameter values are used in eONE-L1 and rONE-L1.

\textbf{AMP: }terminating if $\frac{\twon{\biAx_t-\bib}}{\twon{\bib}}<10^{-5}$.

\textbf{FPC-AS: }$\lambda=2\times10^{-6}$ and $gtol=1\times10^{-6}$, where $gtol$ is the termination criterion on the maximum norm of sub-gradient. FPC-AS solves the problem $\qp$.

\textbf{NESTA: }$\lambda=2\times10^{-6}$, $\epsilon=0$ and the termination criterion $tolvar=1\times10^{-8}$. NESTA solves $\bpdn$ using the Nesterov algorithm \cite{nesterov2005smooth}, with continuation.

Our experiment results are presented in Table \ref{table_comparison}. In both `easy' and `hard' cases, rONE-L1 is much faster than eONE-L1. In the `easy' case, the proposed rONE-L1 algorithm takes the most number of calls of $\biA$ and $\biA'$, except that of eONE-L1, due to a conservative setting of $r$. But this number of calls (515.4) is very close to that of NESTA (468.9), and furthermore, the CPU time usage of rONE-L1 (2.14 s) is less than that of NESTA (2.70 s) because of its concise implementation. In the `hard' case, rONE-L1 has the second best performance with significantly less CPU time than that of AMP and NESTA. AMP has the second worst CPU time and the worst accuracy as the dynamic threshold in each iteration depends on the mean squared error of the current iterative solution, which cannot be calculated exactly in the implementation.
\begin{table}
 \caption{Comparison Results of ONE-L1 Algorithms With State-of-the-art Methods}
 \centering
\footnotesize
\begin{tabular}{@{}c|l|l@{ }l|l@{ }l|l@{ }l@{}}
  \hline\hline
  $\rho$ & Method & \multicolumn{2}{c}{\# calls $\biA$ \& $\biA'$}& \multicolumn{2}{|c|}{CPU time (s)} & \multicolumn{2}{c}{Error $(10^{-5})$}\\\hline\hline
\multicolumn{1}{c|}{\multirow{5}{0.4cm}{0.1}} & eONE-L1 & 1819 & \scriptsize{(1522,2054)}$^*$ & 5.62 & \scriptsize{(4.67,6.52)} & 0.42 & \scriptsize{(0.11,0.94)} \\
& rONE-L1 & 515.4 & \scriptsize{(286,954)} & 2.14 & \scriptsize{(1.19,3.92)} & 1.08 & \scriptsize{(0.53,1.30)} \\
 & AMP & 222.7 & \scriptsize{(216,234)} & 0.80 & \scriptsize{(0.76,0.86)} & 1.02 & \scriptsize{(0.85,1.15)} \\
 & FPC-AS & 150.2 & \scriptsize{(135,170)} & 0.50 & \scriptsize{(0.44,0.56)} & 1.13 & \scriptsize{(1.07,1.23)} \\
 & NESTA & 468.9 & \scriptsize{(458,484)} & 2.70 & \scriptsize{(2.55,2.98)} & 1.05 & \scriptsize{(0.99,1.13)} \\\hline\hline
\multicolumn{1}{c|}{\multirow{5}{0.4cm}{0.22}} & eONE-L1 & 9038 & \scriptsize{(7270,11194)} & 28.5 & \scriptsize{(22.0,35.8)} & 1.87 & \scriptsize{(0.46,2.66)} \\
& rONE-L1 & 722.3 & \scriptsize{(440,972)} & 2.61 & \scriptsize{(1.63,3.93)} & 1.80 & \scriptsize{(1.37,3.05)} \\
 & AMP & 1708 & \scriptsize{(1150,2252)} & 6.21 & \scriptsize{(4.19,9.11)} & 10.5 & \scriptsize{(6.96,15.8)} \\
 & FPC-AS & 589.4 & \scriptsize{(476,803)} & 2.10 & \scriptsize{(1.65,2.80)} & 1.96 & \scriptsize{(1.46,3.60)} \\
 & NESTA & 1084 & \scriptsize{(890,1244)} & 6.47 & \scriptsize{(5.22,7.50)} & 2.90 & \scriptsize{(1.62,3.98)} \\\hline\hline
\end{tabular}\label{table_comparison}\\
$^*$Three numbers are averaged, minimum and maximum values, respectively.
\end{table}

\subsection{A Practical Example}
This subsection demonstrates the efficiency of rONE-L1 in the general CS where the signal of interest is approximately sparse and measurements are contaminated with noise. We seek to reconstruct the Mondrian image of size $256\times256$, shown in Fig. \ref{fig_comparison_noisy}, from its noise-contaminated partial-DCT coefficients. This image presents a challenge as its wavelet expansion is compressible but not exactly sparse. The sampling pattern, which is inspired by magnetic resonance imaging (MRI) and is shown in Fig. \ref{fig_comparison_noisy}, is adopted as most energy of the image concentrates at low-frequency components after the DCT transform. The measurement vector $\bib$ contains $n=7419$ DCT measurements ($\delta=0.113$). White Gaussian noise with standard deviation $\sigma=1$ is then added. We set $\epsilon=\sqrt{n+2\sqrt{2n}}\sigma$. Haar wavelet with a decomposition level 4 is chosen as the sparsifying transform $\cW$. Hence, the problem to be solved is $\bpdn$ with $\biA=\cF_p\cW'$, where $\cF_p$ is the partial-DCT transform. The reconstructed image is $\hat \biH=\cW'\hat \bix$ with $\hat \bix$ being the reconstructed wavelet coefficients and reconstruction error is calculated as $\frac{\frobn{\hat \biH-\biH^o}}{\frobn{\biH^o}}$, where $\biH^o$ is the original image and $\frobn{\cdot}$ denotes the Frobenius norm. We compare the performance of rONE-L1 with NESTA and FPC-AS.

\begin{rem}
AMP is omitted for its poor performance in this approximately-sparse-signal case. For AMP, the value of the dynamic threshold $\lambda_t$ and the term $\zeron{x_t}$ in (\ref{formula:amp}) depend on the condition that the signal to reconstruct is strictly sparse.
\end{rem}

In such a noisy measurement case, an exact solution for $\bpdn$ is not sought after in the rONE-L1 simulation. The computation of the rONE-L1 algorithm is set to terminate if $\frac{\twon{\biAx_t-\bib}}{\twon{\bib}}\leq\tau=\frac{\epsilon}{\twon{b}}$, i.e., rONE-L1 outputs the first $\bix_t$ when it becomes a feasible solution of $\bpdn$.

\textbf{FPC-AS: }$\lambda=1\times10^{-3}$, $gtol = 1\times10^{-3}$, $gtol\_scale\_x = 1\times10^{-6}$ and the maximum number of iterations for subspace optimization $sub\_mxitr = 10$. The parameters are set according to \cite[Section 4.4]{wen2010fast}.

\textbf{NESTA: }$\lambda=1\times10^{-4}$, and $tolvar=1\times10^{-6}$. The parameters are tuned to achieve the minimum reconstruction error.

Fig. \ref{fig_comparison_noisy} shows the experiment results where rONE-L1, FPC-AS and NESTA produce faithful reconstructions of the original image. The rONE-L1 algorithm produces a reconstruction error (0.0741) lower than that of FPC-AS (0.0809) with comparable computation times (11.1 s and 11.4 s, respectively). While NESTA results in a slightly lower reconstruction error (0.0649), it incurs about twice more computation time (29.4 s).

%

\begin{figure}
  \centering
  \hfill\includegraphics[height=1.1in]{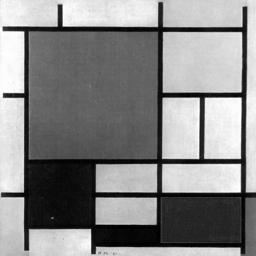}%
  \hfill\includegraphics[height=1.1in]{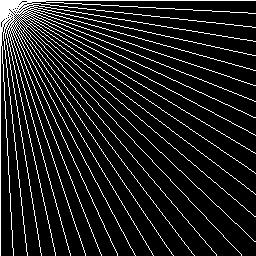}\hspace*{\fill}\vspace*{10pt}

  \hfill\includegraphics[height=1.1in]{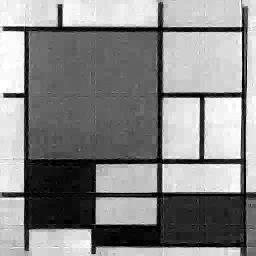}%
  \hfill\includegraphics[height=1.1in]{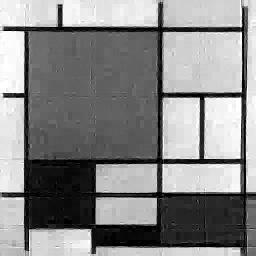}%
  \hfill\includegraphics[height=1.1in]{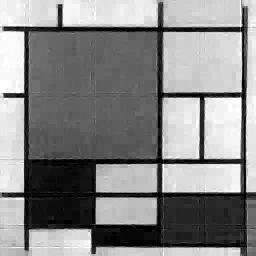}\hspace*{\fill}
  \caption{An example of 2D image reconstruction from noise-contaminated partial-DCT measurements. Upper left: original Mondrian image; upper right: sampling pattern. The lower three are reconstructed images respectively by rONE-L1 (lower left, error: 0.0741, time: 11.1 s), FPC-AS (lower middle, error: 0.0809, time: 11.4 s) and NESTA (lower right, error: 0.0649, time: 29.4 s).}
  \label{fig_comparison_noisy}
\end{figure}

\section{Conclusion}\label{section_conclusion}

In this work, we have presented novel algorithms to solve the basis pursuit problem for noiseless CS. The proposed rONE-L1 algorithm, based on the augmented Lagrange multiplier method and heuristic simplification, can be considered as a modified IST with an aggressive continuation strategy. The following two cases of CS have been studied: 1) exact reconstruction of sparse signals from noise-free measurements, and 2) reconstruction of approximately sparse signals from noise-contaminated measurements. The proposed rONE-L1 outperforms AMP, which is a well-known IST type algorithm, in Case 2 and also in Case 1 when the setting of $\sbra{\delta,\rho}$ is close to the phase transition of basis pursuit. It is faster than NESTA in both Case 1 and Case 2. The numerical experiments further show that rONE-L1 outperforms FPC-AS in Case 2. Apart from the high computational efficiency and accurate reconstruction result, another useful property of rONE-L1 is its ease of parameter tuning. It only needs to set a termination criterion $\tau$ if the recommended $r$ is used and the value of $\tau$ is explicit in Case 2. While this correspondence focuses on reconstruction of real-valued signals, it is straightforward to apply ONE-L1 algorithms to the reconstruction of complex-valued signals \cite{yang2011sparsity}. More rigorous analysis of rONE-L1 is currently under investigation.

\section*{Acknowledgment}
The authors are grateful to the anonymous reviewers for helpful comments. Z. Yang wishes to thank Arian Maleki for providing the AMP codes. The Matlab codes of ONE-L1 algorithms are available at \emph{http://sites.google.com/site/zaiyang0248}.

\bibliographystyle{IEEEtran}
\bibliography{reference}

\begin{thebibliography}{10}
\providecommand{\url}[1]{#1}
\csname url@samestyle\endcsname
\providecommand{\newblock}{\relax}
\providecommand{\bibinfo}[2]{#2}
\providecommand{\BIBentrySTDinterwordspacing}{\spaceskip=0pt\relax}
\providecommand{\BIBentryALTinterwordstretchfactor}{4}
\providecommand{\BIBentryALTinterwordspacing}{\spaceskip=\fontdimen2\font plus
\BIBentryALTinterwordstretchfactor\fontdimen3\font minus
  \fontdimen4\font\relax}
\providecommand{\BIBforeignlanguage}[2]{{%
\expandafter\ifx\csname l@#1\endcsname\relax
\typeout{** WARNING: IEEEtran.bst: No hyphenation pattern has been}%
\typeout{** loaded for the language `#1'. Using the pattern for}%
\typeout{** the default language instead.}%
\else
\language=\csname l@#1\endcsname
\fi
#2}}
\providecommand{\BIBdecl}{\relax}
\BIBdecl

\bibitem{candes2006robust}
E.~Cand{\`e}s, J.~Romberg, and T.~Tao, ``{Robust uncertainty principles: Exact
  signal reconstruction from highly incomplete frequency information},''
  \emph{IEEE Trans. Info. Theo.}, vol.~52, no.~2, pp. 489--509, 2006.

\bibitem{candes2007sparsity}
E.~Cand{\`e}s and J.~Romberg, ``{Sparsity and incoherence in compressive
  sampling},'' \emph{Inverse Problems}, vol.~23, pp. 969--985, 2007.

\bibitem{donoho2006compressed}
D.~Donoho, ``{Compressed sensing},'' \emph{IEEE Transactions on Information
  Theory}, vol.~52, no.~4, pp. 1289--1306, 2006.

\bibitem{daubechies2004iterative}
I.~Daubechies, M.~Defrise, and C.~De~Mol, ``{An iterative thresholding
  algorithm for linear inverse problems with a sparsity constraint},''
  \emph{Communications on Pure and Applied Mathematics}, vol.~57, no.~11, pp.
  1413--1457, 2004.

\bibitem{bredies2008linear}
K.~Bredies and D.~Lorenz, ``{Linear convergence of iterative
  soft-thresholding},'' \emph{Journal of Fourier Analysis and Applications},
  vol.~14, no.~5, pp. 813--837, 2008.

\bibitem{kim2008interior}
S.~Kim, K.~Koh, M.~Lustig, S.~Boyd, and D.~Gorinevsky, ``{An Interior-Point
  Method for Large-Scale $\ell_1$-Regularized Least Squares},'' \emph{Selected
  Topics in Signal Processing, IEEE Journal of}, vol.~1, no.~4, pp. 606--617,
  2008.

\bibitem{lustig2007sparse}
M.~Lustig, D.~Donoho, and J.~Pauly, ``{Sparse MRI: The application of
  compressed sensing for rapid MR imaging},'' \emph{Magnetic Resonance in
  Medicine}, vol.~58, no.~6, pp. 1182--1195, 2007.

\bibitem{hale2007fixed}
E.~Hale, W.~Yin, and Y.~Zhang, ``{A fixed-point continuation method for
  l1-regularized minimization with applications to compressed sensing},''
  \emph{CAAM TR07-07, Rice University}, 2007.

\bibitem{candes-l1}
E.~Cand{\`e}s and J.~Romberg, ``{$\ell1$-magic: Recovery of sparse signals via
  convex programming},'' \emph{URL:
  www.acm.caltech.edu/l1magic/downloads/l1magic.pdf}.

\bibitem{becker2009nesta}
S.~Becker, J.~Bobin, and E.~Candes, ``{NESTA: A fast and accurate first-order
  method for sparse recovery},'' \emph{SIAM J. Imaging Sciences}, vol.~4,
  no.~1, pp. 1--39, 2011.

\bibitem{nesterov2005smooth}
Y.~Nesterov, ``{Smooth minimization of non-smooth functions},''
  \emph{Mathematical Programming}, vol. 103, no.~1, pp. 127--152, 2005.

\bibitem{tropp2007signal}
J.~Tropp and A.~Gilbert, ``{Signal recovery from random measurements via
  orthogonal matching pursuit},'' \emph{IEEE Transactions on Information
  Theory}, vol.~53, no.~12, pp. 4655--4666, 2007.

\bibitem{donoho2006sparse}
D.~Donoho, Y.~Tsaig, I.~Drori, and J.~Starck, ``{Sparse solution of
  underdetermined linear equations by stagewise orthogonal matching pursuit},''
  \emph{submitted to IEEE Trans. on Signal Processing. Available at:
  www.cs.tau.ac.il/$\sim$idrori/StOMP.pdf}, 2006.

\bibitem{needell2009cosamp}
D.~Needell and J.~Tropp, ``{CoSaMP: Iterative signal recovery from incomplete
  and inaccurate samples},'' \emph{Applied and Computational Harmonic
  Analysis}, vol.~26, no.~3, pp. 301--321, 2009.

\bibitem{wen2010fast}
Z.~Wen, W.~Yin, D.~Goldfarb, and Y.~Zhang, ``{A fast algorithm for sparse
  reconstruction based on shrinkage, subspace optimization and continuation},''
  \emph{SIAM Journal on Scientific Computing}, vol.~32, no.~4, pp. 1832--1857,
  2010.

\bibitem{maleki2010optimally}
A.~Maleki and D.~Donoho, ``{Optimally tuned iterative thresholding algorithms
  for compressed sensing},'' \emph{IEEE J. Select. Areas Signal Processing},
  vol.~4, pp. 330--341, 2010.

\bibitem{donoho2009message}
D.~Donoho, A.~Maleki, and A.~Montanari, ``{Message-passing algorithms for
  compressed sensing},'' \emph{Proceedings of the National Academy of
  Sciences}, vol. 106, no.~45, pp. 18\,914--18\,919, 2009.

\bibitem{nocedal1999numerical}
J.~Nocedal and S.~Wright, \emph{{Numerical Optimization}}.\hskip 1em plus 0.5em
  minus 0.4em\relax New York: Springer verlag, 2006.

\bibitem{bertsekas1982constrained}
D.~Bertsekas, \emph{{Constrained optimization and Lagrange multiplier
  methods}}.\hskip 1em plus 0.5em minus 0.4em\relax Boston: Academic Press,
  1982.

\bibitem{donoho2010counting}
D.~Donoho and J.~Tanner, ``{Counting the faces of randomly-projected hypercubes
  and orthants, with applications},'' \emph{Discrete and Computational
  Geometry}, vol.~43, no.~3, pp. 522--541, 2010.

\bibitem{stojnic2009various}
M.~Stojnic, ``{Various thresholds for l1-optimization in compressed sensing},''
  \emph{arXiv:0907.3666}, 2009.

\bibitem{donoho2009observed}
D.~Donoho and J.~Tanner, ``{Observed universality of phase transitions in
  high-dimensional geometry, with implications for modern data analysis and
  signal processing},'' \emph{Philosophical Transactions of the Royal Society
  A}, vol. 367, no. 1906, pp. 4273--4293, 2009.

\bibitem{yang2011sparsity}
Z.~Yang and C.~Zhang, ``{Sparsity-undersampling tradeoff of compressed sensing
  in the complex domain},'' in \emph{Proceedings of ICASSP 2011}, 2011.

\end{thebibliography}

\end{document}